\documentclass[11pt,letterpaper]{article}
\usepackage{amsmath}

\usepackage{comment}
\usepackage{amssymb}
\usepackage{fullpage}
\usepackage{color}
\usepackage{amsthm}

\usepackage{url}
\usepackage{xspace}
\usepackage{cleveref}

\newtheorem{theorem}{Theorem}
\newtheorem{theorem*}{Theorem}
\newtheorem{corollary}[theorem]{Corollary}
\newtheorem{lemma}[theorem]{Lemma}

\newtheorem{proposition}[theorem]{Proposition}
\newtheorem*{proposition*}{Proposition}
\newtheorem{definition}[theorem]{Definition}

\newtheorem{fact}[theorem]{Fact}

\theoremstyle{definition}

\newcommand{\F}{\mathbb{F}}
\newcommand{\N}{\mathbb{N}}
\newcommand{\Q}{\mathbb{Q}}
\newcommand{\Z}{\mathbb{Z}}
\newcommand{\R}{\mathbb{R}}

\newcommand{\claimproof}[2]%
{\noindent{{\bf Proof of Claim} \ref{#1}:}
#2\hspace*{\fill}$\Box$~~~~~\vspace{5mm} }

\renewcommand{\mod}{\mbox{mod~}}

\newcommand{\vecs}{{\mathbf{s}}}
\newcommand{\vecu}{{\mathbf{u}}}
\newcommand{\vecv}{{\mathbf{v}}}
\newcommand{\vecw}{{\mathbf{w}}}
\newcommand{\vecx}{{\mathbf{x}}}

\newcommand{\veca}{{\mathbf{a}}}

\newcommand{\vecc}{{\mathbf{c}}}
\newcommand{\vecd}{{\mathbf{d}}}
\newcommand{\vece}{{\mathbf{e}}}

\newcommand{\VP}{\text{VP}}
\newcommand{\VNP}{\mathrm{VNP}}

\newcommand{\CountP}{\text{\#P}}

\renewcommand{\mod}{\text{mod }}

\newcommand{\bcover}{\mathrel{\dot{<}}}

\newcommand{\poly}{\mathsf{poly}}
\renewcommand{\F}{\mathbb{F}}

\newcommand{\dom}{\triangleright}

\newcommand{\aperm}[1]{\langle #1\rangle}
\newcommand{\acode}[1]{\mathfrak{c}(#1)}
\newcommand{\coeff}{\mathrm{Coeff}}
\newcommand{\ks}{\mathrm{KS}}
\newcommand{\bbox}{\mathcal{B}}

\newcommand{\schur}{\mathrm{s}}
\newcommand{\schub}{Y}

\newcommand{\oracle}{\mathcal{O}}

\definecolor{edcolor}{rgb}{0,0.8,0.3}

\newcommand{\Anote}[1]{}
\newcommand{\Nnote}[1]{}
\newcommand{\Ynote}[1]{}

\begin{document}

\title{Sparse multivariate polynomial interpolation in the basis of Schubert 
polynomials}

\author{Priyanka Mukhopadhyay \thanks{Center for Quantum Technologies, National 
University of Singapore ({\tt mukhopadhyay.priyanka@gmail.com}). }
\and Youming Qiao \thanks{Centre for Quantum Computation and Intelligent Systems, 
 University of Technology, Sydney, Australia
 ({\tt jimmyqiao86@gmail.com}).}
}

\maketitle

\begin{abstract}
  Schubert polynomials were discovered by A. Lascoux and M. Sch\"utzenberger in the 
study of cohomology rings of flag manifolds in 1980's. These polynomials 
generalize Schur polynomials, and form a linear basis of multivariate polynomials. 
In 2003, Lenart and Sottile introduced skew Schubert polynomials, which generalize 
skew Schur polynomials, and expand in the Schubert basis with the generalized 
Littlewood-Richardson coefficients. 

In this paper we initiate the study of these two families of polynomials from the 
perspective of computational complexity theory. 
We first observe that skew Schubert polynomials, and therefore Schubert 
polynomials, are in $\CountP$ (when evaluating on non-negative integral inputs) and 
$\VNP$.

Our main result is a deterministic algorithm that computes the expansion of 
a polynomial $f$ of degree $d$ in $\Z[x_1, \dots, x_n]$ in the basis of Schubert 
polynomials, 
assuming an oracle computing Schubert polynomials. This algorithm runs in time 
polynomial in $n$, $d$, and the bit size of the expansion. This 
generalizes, and derandomizes, the sparse interpolation algorithm of symmetric 
polynomials in the Schur basis by 
Barvinok and Fomin (Advances in Applied Mathematics, 18(3):271--285). 
In fact, our interpolation 
algorithm is general enough to accommodate any linear basis satisfying certain 
natural properties. 

Applications of the above results include a new algorithm that computes the 
generalized Littlewood-Richardson coefficients. 

\end{abstract}






\section{Introduction}\label{sec:intro}

\paragraph{Polynomial interpolation problem.} The classical polynomial 
interpolation problem starts with a set of data points, 
$(a_1, b_1)$, \dots, $(a_{n+1},$ $b_{n+1})$, where $a_i, b_i\in\Q$, $a_i\neq a_j$ 
for $i\neq j$, and asks for 
a univariate polynomial $f\in \Q[x]$ s.t. $f(a_i)=b_i$ for $i=1, \dots, n+1$. 
While 
such a polynomial of degree $\leq n$ exists and is unique, depending on different 
choices of linear bases, several formulas, under the name of Newton, Lagrange, 
and Vandermonde, have become classical. 

In theoretical computer science (TCS), the following problem also bears the name 
interpolation of polynomials, studied from 1970's: 
given a black-box access to a multivariate polynomial $f\in \F[x_1, \dots, x_n]$ 
($\F$ a field), compute $f$ by writing out its sum-of-monomials expression. 
Several algorithms have been 
proposed to solve this problem \cite{Zip79,BT88,KY88,Zip90,ks01}. A natural 
generalization is to consider expressing $f$ using 
the more powerful arithmetic circuits. In this more general setting, 
the problem is called the reconstruction problem for arithmetic circuits 
\cite[Chap. 5]{sy2010}, and the sum-of-monomials 
expression of $f$ is viewed as a subclass of arithmetic circuits, namely depth-2 
$\Sigma \Pi$ circuits. Reconstruction problems for various models 
gained quite momentum recently  
\cite{bc98,shpilka07,sv08,ks09,ams10,gkl11,kay12,gkl2012,GKQ14}.

As mentioned, for interpolation of univariate polynomials, different formulas 
depend on different choices of linear bases. On the other hand, for interpolation 
(or more precisely, reconstruction) of multivariate polynomials in the TCS 
setting, the algorithms depend on the computation models 
crucially. In the latter 
context, to our best knowledge, only the linear basis of monomials (viewed as 
depth-2 $\Sigma \Pi$ circuits) has been considered. 

\paragraph{Schubert polynomials.} In this 
paper, we consider the interpolation of multivariate polynomials in the 
TCS 
setting, but in another linear basis of multivariate polynomials, namely the 
Schubert polynomials. This provides another natural direction for 
generalizing the multivariate polynomial interpolation problem. Furthermore, as 
will be explained below, such an interpolation algorithm can be used to compute 
certain quantities in 
geometry that are of great interest, yet not well-understood. 

Schubert polynomials were discovered by Lascoux and Sch\"utzenberger \cite{LS82} 
in the study of cohomology rings of flag manifolds in 1980's. See \Cref{def:skew} 
or \Cref{def:schub} for the definition\footnote{\Cref{def:schub} is one of the 
classical definitions of 
Schubert polynomials, while \Cref{def:skew} defines Schubert polynomials in the 
context of skew Schubert polynomials. }, and 
\cite{MD_schubert,Manivel} for detailed results. For now we only 
point out that (1) Schubert polynomials in $\Z[x_1, 
\dots, x_n]$ are indexed by $\vecv=(v_1, \dots, v_n)\in \N^n$, denoted by 
$\schub_{\vecv}$;\footnote{In the literature, it is more common that Schubert 
polynomials indexed by permutations instead of $\N^n$. These two index sets are 
equivalent, through the correspondence between 
permutations and $\N^n$ as described in \Cref{sec:prel}. We adopt $\N^n$ in the introduction
because they are easier to work with when dealing with a fixed number of 
variables.} (2) $\schub_{\vecv}$ is homogeneous of degree 
$\sum_{i=1}^nv_i$.\footnote{Algorithms in this work run 
in time polynomial in the degree of the polynomial. By (2) this is equivalent to 
that the indices of Schubert polynomials are given in unary.}

Schubert polynomials have many distinguished properties. They form a linear basis 
of multivariate polynomials, and yield a generalization of Newton 
interpolation formula to the multivariate case \cite[Sec. 9.6]{Lascoux_symmetric}. 
Also, Schur polynomials are special Schubert polynomials (\Cref{fact:schub} 
(1)). 
A Schubert polynomial can contain exponentially many monomials: the 
complete homogeneous symmetric polynomials are special Schubert polynomials 
(\Cref{fact:schub} (2)).
It is not clear to us whether Schubert polynomials have polynomial-size 
arithmetic circuits. Because of these reasons, interpolation in the Schubert basis 
could not be covered by the aforementioned results for the reconstruction 
problems, unless the
arithmetic circuit complexity of Schubert polynomials is understood better: at 
present, we are only able to put Schubert polynomials in $\VNP$.

While Schubert polynomials are mostly studied due to their deep geometric 
meanings (see e.g. \cite{KM05}), they do have certain algorithmic aspects that 
have been studied shortly after their introduction in 1982. Indeed, an early paper 
on Schubert polynomials by Lascoux and Sch\"utzenberger was concerned about using 
them to compute the Littlewood-Richardson coefficients \cite{LS85}. 
That procedure has been implemented in the program system {\tt Symmetrica} 
\cite{KKL92}, which includes a set of routines to work with Schubert 
polynomials. On the other hand, the complexity-theoretic study of the algorithmic 
aspects of Schubert polynomials seems lacking, and we hope that this paper serves 
as a modest 
step towards this direction. 

\paragraph{Our results.} Our main result is about deterministic interpolation of 
sparse polynomials with integer 
coefficients in the Schubert basis, modulo an oracle that computes Schubert 
polynomials. The complexity is measured by the bit size of the representation.
\begin{theorem}\label{thm:interpolate}
Suppose we are given (1) black-box access to some polynomial $f\in\Z[x_1, \dots, 
x_n]$, $f=\sum_{\vecv\in \Gamma} a_{\vecv} \schub_{\vecv}$, $a_\vecv\neq 0\in 
\Z$ with the promise that $\deg(f)\leq d$ and $|\Gamma|\leq m$; (2) an oracle that 
computes the evaluation of Schubert polynomials on nonnegative
integral points. Then there exists a deterministic algorithm, that outputs the 
expansion of $f$ in the basis of Schubert polynomials. The algorithm runs in time 
polynomial in $n$, $d$, $m$, and $\log(\sum_{\vecv\in \Gamma}|a_{\vecv}|)$.
\end{theorem}

In fact, \Cref{thm:interpolate} relies on the algorithm in 
\Cref{thm:interpolate_good} which applies to a more general setting: it only 
requires the linear basis where the leading monomials are ``easy to isolate.'' See 
\emph{Our techniques} for a more detailed discussion. 

\Cref{thm:interpolate} generalizes and derandomizes a result by Barvinok 
and Fomin, who in \cite{BF97} present a randomized algorithm 
that interpolates sparse symmetric polynomials in the Schur basis. As mentioned, 
Schur polynomials are special Schubert polynomials, and the Jacobi-Trudi formulas 
for Schur polynomials yield efficient 
algorithms to compute them. So to recover Barvinok and Fomin's result, 
apply our interpolation algorithm to symmetric 
polynomials, and replace the $\CountP$ oracle computing Schubert polynomials by 
the efficient algorithm computing Schur polynomials. Likewise, for those Schubert 
polynomials with efficient evaluation procedures\footnote{For example, there are determinantal formulas \cite[Sec. 2.6]{Manivel} for 
Schubert polynomials indexed by 2143-avoiding permutations ($\nexists 
i<j<k<\ell$ s.t. $\sigma(j)<\sigma(i)<\sigma(\ell)<\sigma(k)$), and 321-avoiding 
permutations ($\nexists i<j<k$ s.t. $\sigma(k)<\sigma(j)<\sigma(i)$). 
2143-avoiding permutations are also known as vexillary permutations and form a 
generalization of Grassmannian permutations.}, we 
can get rid of the $\CountP$ oracle to obtain a polynomial-time interpolation 
algorithm. 

Our second result concerns the evaluation of Schubert polynomials. In fact, we 
shall work with a generalization of Schubert polynomials, namely \emph{skew 
Schubert polynomials} as defined by Lenart and Sottile \cite{LS03}. We will 
describe the definition in \Cref{sec:prel}. For now, we only remark that skew 
Schubert polynomials generalize Schubert polynomials in a way analogous to how 
skew Schur polynomials generalize Schur polynomials. A skew Schubert polynomial, 
denoted by
$\schub_{\vecw/\vecv}$, is 
indexed by $\vecv\leq \vecw\in \N^m$ where $\leq$ denotes the Bruhat 
order\footnote{Bruhat order on codes is inherited from the Bruhat order on 
permutations through the correspondence between codes and permutations as in \Cref{sec:prel}.}. 
Schubert polynomials 
can be defined by setting $\vecw$ to correspond to the permutation maximal in the 
Bruhat order.

\begin{theorem}\label{thm:eval}
Given $\vecv, \vecw\in\N^m$ in unary, and $\veca\in\N^n$ in binary, 
computing $\schub_{\vecw/\vecv}(\veca)$ is in $\CountP$.
\end{theorem}

\begin{corollary}\label{cor:eval}
Given $\vecv\in\N^n$ in unary, and $\veca\in\N^n$ in binary, 
computing $\schub_{\vecv}(\veca)$ is in $\CountP$.
\end{corollary}

Note that in \Cref{thm:eval} we have $\vecv, \vecw\in \N^m$, while in 
\Cref{cor:eval} we have $\vecv\in \N^n$. This is because, while for 
Schubert polynomials we know $\schub_\vecv$ for $\vecv\in\N^n$ depends on $n$ 
variables, such a relation is not clear for skew Schubert polynomials. 

Finally, we also study these polynomials in the framework of algebraic complexity. 

\begin{theorem}\label{thm:vnp}
Skew Schubert polynomials, and therefore Schubert polynomials, are in $\VNP$.
\end{theorem}

\paragraph{Applications of our algorithms.} 
A long-standing open problem about Schubert polynomials is to give a 
combinatorially positive 
rule, or in other words, a $\CountP$ algorithm for the \emph{generalized 
Littlewood-Richardson (LR) coefficients}, defined as the 
coefficients of the 
expansion of products of two Schubert polynomials in the Schubert basis. They are 
also the coefficients of the expansion of skew Schubert polynomials in the 
Schubert basis \cite{LS03}. These numbers are of great interest in algebraic 
geometry, since they are the intersection numbers of Schubert varieties in the 
flag manifold. See \cite{ABS14,MPP14} for recent developments on this.

The original LR coefficients are a special case when replacing Schubert with Schur 
in the above definition. It is known that the original LR coefficients are 
$\CountP$-complete, by the celebrated LR rule, and a result of 
\cite{Nara06}. Therefore the generalized LR coefficients are also $\CountP$-hard 
to 
compute, while as mentioned, putting this problem in $\CountP$ is considered to be 
very difficult -- in fact, we were not aware of any non-trivial complexity-theoretic upper 
bound. On the other hand, by interpolating skew Schubert polynomials in the 
Schubert basis, we have the following.

\begin{corollary}\label{cor:psharpp}
Given $\vecw, \vecv\in \N^m$, let $\Gamma=\{\vecu\in \N^m \mid a^{\vecv, 
\vecu}_{\vecw}\neq 0\}$. Then there exists a deterministic algorithm that, given 
access to an $\CountP$ oracle, 
computes $(a^{\vecv, \vecu}_{\vecw} | \vecu\in \N^m)$ in time polynomial in 
$|\vecw|, |\Gamma|$, and $\log(\sum_{\vecu\in \Gamma} a^{\vecv, 
\vecu}_{\vecw})$. 
\end{corollary}

The algorithm in \Cref{cor:psharpp} has the benefit of running in time polynomial 
in the 
\emph{bit size} of $a^{\vecv, \vecu}_{\vecw}$. 
Therefore, when $|\Gamma|$ is small compared to $\sum_\vecw a^{\vecv,\vecu}_\vecw$, 
our algorithm is expected to lead to a notable saving, compared to those algorithms that are solely based on positive rules, 
e.g. \cite{Kogan01}. (\cite{Kogan01} furthermore only deals with the case of 
Schubert times Schur.) 
Of course, in practice we need to take into account the time for evaluating Schubert polynomials. 

In addition, we note that Barvinok and Fomin's original motivation is to compute e.g. the Littlewood-Richardson coefficients, 
Kostka numbers, and the irreducible characters of the symmetric group. 
See \cite[Sec. 1]{BF97} for the definitions and importance of these numbers. 
Since our algorithm can recover theirs (without referring to a $\CountP$ oracle), it can be used to compute these quantities as well. 
Note that our algorithm is moreover deterministic. 

Our original motivation of this work was to better understand this approach of Barvinok and Fomin to compute the 
LR coefficients. This topic recently receives attention in complexity theory \cite{Nara06,MNS12,BI13_positivity}, due to 
its connection to the geometric complexity theory (GCT) \cite{GCT_JACM,MNS12}. 
Though this direction of generalization does not apply to GCT directly, we believe it helps in a better understanding 
(e.g. a derandomization) of this approach of computing the LR coefficients.

\paragraph{Our techniques.} We achieve \Cref{thm:interpolate} by first formalizing some natural properties of a 
linear basis of (subrings of) multivariate polynomials (\Cref{subsec:good}).
These are helpful for the interpolation purpose. 
If a basis satisfies these properties, we call this basis \emph{interpolation-friendly}, or I-friendly for short. 
Then we present a deterministic interpolation algorithm for I-friendly bases (\Cref{thm:interpolate_good}). 
We then prove that the Schubert basis is interpolation-friendly\footnote{While the proofs 
for these properties of Schubert polynomials are easy, and should be 
known by experts, we include complete proofs as we could not find 
complete proofs or explicit statements. 
Most properties of Schubert polynomials, e.g. the 
$\CountP$ result, \Cref{lem:dom}, \Cref{prop:schub_bound}, can also 
be obtained by a combinatorial tool called RC graphs \cite{BB93}. We 
do not attempt to get optimal results (e.g. in \Cref{cor:long_trans} and 
\Cref{prop:schub_bound}), but are content with bounds that are good 
enough for our purpose. } (\Cref{sec:schub_good}).

Technically, for the interpolation algorithm, we combine the structure 
of the Barvinok-Fomin algorithm with several ingredients from the 
elegant deterministic interpolation algorithm for sparse polynomials in the 
monomial basis by Klivans and Spielman \cite{ks01}. We deduce the key property for 
Schubert polynomials to be I-friendly, via the transition formula of Lascoux 
and Sch\"utzenberger \cite{LS85}. The concept of I-friendly bases and 
the corresponding interpolation algorithm may be of independent interest, since 
they may be used to apply to other bases of (subrings of) the multivariate 
polynomial 
ring, e.g. Grothendieck polynomials and Macdonald polynomials 
\cite{Lascoux_polynomial}. 

We would like to emphasize a subtle point of Schubert polynomials 
that is cruicial for our algorithm: for $\schub_{\vecv}$, if the monomial 
$\vecx^\vecu$ is in $\schub_{\vecv}$, then $\vecv$ dominates $\vecu$ reversely; 
that 
is, $v_n\geq u_n$, $v_n+v_{n-1}\geq u_n+u_{n-1}$, \dots, $v_n+\dots+v_1\geq 
u_n+\dots +u_1$. While by no means a difficult property and known to 
experts, it is interesting 
that the only reference we can find is a footnote in Lascoux's book \cite[pp. 62, 
footnote 4]{Lascoux_polynomial}, so we prove it in \Cref{lem:dom}. On 
the other hand, in the literature a weaker property is often mentioned, that is 
$\vecv$ is 
no less than $\vecu$ in the reverse lexicographic order.
However, this order turns out to \emph{not} suffice for the purpose of interpolation. 

\paragraph{Comparison with the Barvinok-Fomin algorithm.}

The underlying 
structures of the Barvinok-Fomin algorithm and ours are quite similar. There are 
certain major differences though. 

From the mathematical side, note that Barvinok 
and Fomin used the dominance order of monomials, which corresponds to the use of 
upper triangular matrix in \Cref{subsec:good}. On the other hand, we make 
use of the reverse dominance order, which corresponds to the use of lower 
traingular matrix in \Cref{prop:schub_friendly}. 
It is not hard to see that the dominance order could not work for all Schubert polynomials. 
We also need to upgrade several points (e.g. the computation, and the bounds on coefficients) from Schur polynomials to Schubert polynomials.

From the algorithmic side, both our algorithm and the Barvinok-Fomin algorithm 
reduce multivariate interpolation to univariate interpolation. 
Here are two key differences. Firstly, Barvinok and Fomin relied on randomness to 
obtain a set of linear forms 
s.t. most of them achieve distinct values for a small set of vectors. We resort to 
a deterministic construction of Klivans and Spielman for this set, therefore 
derandomizing the Barvinok-Fomin algorithm. Secondly, our 
algorithm has a recursive structure as the Barvinok-Fomin algorithm. But in each 
recursive step, the approaches are different; ours is based on the method of the 
Klivans-Spielman algorithm.
As a consequence, our
algorithm does not need to know the bounds on the coefficients in the 
expansion, while the basic algorithm in \cite[Sec. 4.1]{BF97} does. Barvinok 
and Fomin avoided the dependence on this bound via binary search and probabilistic 
verification in \cite[Sec. 4.2]{BF97}. However, it seems difficult to derandomize 
this probabilistic verification procedure. 

\paragraph{Organization.} In \Cref{sec:prel} we present certain 
preliminaries. In \Cref{sec:sharp_p} 
we define skew Schubert polynomials and Schubert polynomials, and present the 
proof for \Cref{thm:eval}, \Cref{cor:eval} and \Cref{thm:vnp}. 
In \Cref{sec:interpolate} we define interpolation-friendly bases, and present the 
interpolation algorithm in such bases. 
In \Cref{sec:schub_good} we prove that Schubert polynomials form an 
I-friendly basis, therefore proving \Cref{thm:interpolate}. We remind the reader 
that skew Schubert polynomials are only studied in \Cref{sec:sharp_p}. 


\section{Preliminaries}\label{sec:prel}

\paragraph{Notations.} For $n\in \N$, $[n]:=\{1, \dots, n\}$.
Let $\vecx=(x_1, \dots, x_n)$ be a tuple of $n$ variables. When no ambiguity, 
$\vecx$ may represent the set $\{x_1, \dots, x_n\}$. 
For $\vece=(e_1, \dots, e_n)\in \N^n$, the monomial with exponent $\vece$ 
is $\vecx^\vece:=x_1^{e_1}\dots x_n^{e_n}$.
Given $f\in \Z[\vecx]$ and $\vece\in\N^n$, $\coeff(\vece, f)$ denotes 
the coefficient of $\vecx^\vece$ in $f$. $\vecx^\vece$ (or $\vece$) is in $f$ if 
$\coeff(\vece, f)\neq 0$, and $E_f:=\{\vece\in \N^n\mid 
\vecx^\vece\in f\}$. 
Given two vectors $\vecc=(c_1, \dots, c_n)$ and $\vece=(e_1, \dots, e_n)$ in 
$\Q^n$, their inner product is $\langle \vecc, \vece\rangle=\sum_{i=1}^n c_i e_i$. 
Each $\vecc\in \Q^n$ defines a linear form $\vecc^\star$, which maps 
$\vece\in\Q^n$ to $\langle \vecc, \vece\rangle$.

\paragraph{Codes and permutations.} We call $\vecv=(v_1, 
\dots, v_n)\in \N^n$ a \emph{code}. We identify $\vecv$ with $\vecv'=(v_1, \dots, 
v_n, 0, \dots, 0)\in \N^m$, for $m\geq n$. The weight of the code $\vecv$, denoted 
by $|\vecv|$, is $\sum_{i=1}^n v_i$. A code $\vecv\in\N^n$ is \emph{dominant} if 
$v_1\geq v_2\geq \dots\geq v_n$. We define a \emph{partition} to be a dominant 
code, and often represent it using $\alpha$.
$\vecv$ is \emph{anti-dominant} if $v_1\leq v_2\leq \dots \leq v_k$ and 
$v_{k+1}=\dots=v_n=0$. 

For $N\in \N$, $S_N$ is the symmetric group on $[N]$. A permutation 
$\sigma\in 
S_N$ is written as
$\underline{\sigma(1),\sigma(2), \dots, \sigma(N)}$. We identify $\sigma$ with 
$\sigma'=\underline{\sigma(1), \dots, \sigma(N), N+1, \dots, M}\in S_M,$
for $M\geq N$. The length of $\sigma$, denoted by $|\sigma|$, is the number of 
inversions of $\sigma$.
That is, $|\sigma| = | \{ (i,j) : i < j; \sigma(i) > \sigma(j) \} |$.

Given a permutation $\sigma\in S_N$, we can associate a code 
$\vecv\in \N^n$,\footnote{This is known as the \emph{Lehmer code} of 
a permutation; see e.g. \cite[Section 2.1]{Manivel}.} by assigning 
$v_i=|\{j: j>i, \sigma(j)<\sigma(i)\}|$.
On the other hand, given a code $\vecv \in \N^n$ we 
associate a permutation $\sigma\in S_N (N \geq n)$ as follows. ($N$ will be clear from the 
construction procedure.) To start, $\sigma(1)$ is assigned as $v_1+1$. $\sigma(2)$ 
is the $(v_2+1)$th number, skipping $\sigma(1)$ if necessary (i.e. if $\sigma(1)$ is within the first $(v_2 + 1)$ numbers). 
$\sigma(k)$ is then the $(v_k+1)$th number, skipping some of $\sigma(1), \dots, \sigma(k-1)$ if 
necessary. 
For example, it can be verified that $\underline{316245}$ gives the code $(2, 0, 3, 0, 0, 0)=(2, 0, 3)$ and vice versa. 

Given a code $\vecv$ its associated permutation is denoted as $\langle 
\vecv\rangle$. Conversely, the code of a permutation $\sigma\in S_N$ is denoted as 
$\acode{\sigma}$. It is clear that $|\sigma|=|\vecv|$.

\paragraph{Bases.} Let $R$ be a (possibly nonproper) 
subring of the polynomial ring 
$\Z[\vecx]$. Suppose $M$ is a basis of $R$ as a $\Z$-module. 
$M$ is usually indexed by some \emph{index set} $\Lambda$, and 
$M=\{t_{\lambda}\mid \lambda\in\Lambda\}$. 
For example $\Lambda=\N^n$ for $R=\Z[\vecx]$, and $\Lambda=\{\text{partitions in 
}\N^n\}$ for $R=\{\text{symmetric polynomials}\}$. 
$f\in R$ can be 
expressed uniquely as $f=\sum_{\lambda \in \Gamma} a_{\lambda} t_{\lambda}$, 
$a_{\lambda}\neq 0\in \Z$, $t_{\lambda}\in M$, and a finite $\Gamma\subseteq 
\Lambda$.

\paragraph{A construction of Klivans and Spielman.} We present a 
construction of Klivans and Spielman that is the key to the derandomization here. 
Given positive integers 
$m$, $n$, and $0<\epsilon<1$, let $t=\lceil
m^2n/\epsilon \rceil$. Let $d$ be another positive integer, and fix a prime $p$ 
larger than $t$ and $d$. Now define a set of $t$ vectors in $\N^n$ as $\ks(m, n, 
\epsilon, d, p):=\{\vecc^{(1)}, \dots, \vecc^{(t)}\}$ by
$
\vecc^{(k)}_i=k^{i-1}\ \mod p,
$ for $i\in[n]$. Note that $\vecc^{(k)}_i\leq p=O(m^2nd/\epsilon)$.
The main property we need from $\ks$ is the following.
\begin{lemma}[{\cite[Lemma 3]{ks01}}]\label{lem:ks}
Let $\vece^{(1)}, \dots, \vece^{(m)}$ be $m$ distinct vectors from $\N^n$ with 
entries in 
$\{0, 1, \dots, d\}$. Then 
\begin{eqnarray}
\Pr_{k\in[t]}[\langle \vecc^{(k)}, \vece^{(j)}\rangle \text{ are distinct for } j\in[m] ]\geq 1- m^2n/t\geq 1-\epsilon \nonumber
\end{eqnarray}
\end{lemma}

\paragraph{On $\CountP$ and $\VNP$.} The standard definition 
of $\CountP$ is as follows:
function $f:\cup_{n\in \N}\{0, 1\}^n\to \Z$ is in $\CountP$, if there exists a 
polynomial-time Turing machine $M$ and a polynomial $p$, s.t. for any $x\in\{0, 
1\}^n$, $f(x)=|\{y\in\{0, 1\}^{p(n)}$ $\text{ s.t. } M \text{ accepts } (x, 
y)\}|$. 
In the proof of \Cref{thm:eval} in \Cref{sec:sharp_p}, we find it 
handy to consider the class of Turing machines that 
output a nonnegative integer (instead 
of just accept or reject), and functions $f:\cup_{n\in\N}\{0, 1\}^n\to 
\N$ s.t. there exists a 
polynomial-time Turing machine $M$ and a polynomial $p$, s.t. for any $x\in\{0, 
1\}^n$, $f(x)=\sum_{y\in\{0, 1\}^{p(n)}} M(x, y)$. As described in \cite{CSW13}, 
such functions are also in $\CountP$, as we can construct a usual Turing machine 
$M'$, which takes $3$ arguments $(x, y, z)$, and $M'(x, y, z)$ accepts if and 
only if $z<M(x, y)$. Then $\sum_{y, z}M'(x, y, z)=\sum_y M(x, y)$. Note that $z\in\{0, 1\}^{q(n)}$ for some polynomial $q$ as $M$ is polynomial-time.

The reader is referred to \cite{sy2010} for basic notions like arithmetic 
circuits. $\VP$ denotes the class of polynomial families 
$\{f_n\}_{n\in\N}$ s.t. each $f_n$ is a polynomial in $\poly(n)$ variables, of 
$\poly(n)$ degree, and can be computed by an arithmetic circuit of size 
$\poly(n)$. $\VNP$ is the class of polynomials $\{g_n\}_{n\in\N}$ s.t. 
$g_n(x_1, \dots, x_n)=\sum_{(c_1, \dots, c_m)\in\{0, 1\}^m} f_n(x_1, \dots, x_n, 
c_1, \dots, 
c_m)$ where $m=\poly(n)$, and $\{f_n\}_{n\in\N}$ is in $\VP$. Valiant's criterion 
is useful to put polynomial families in $\VNP$. 

\begin{theorem}[Valiant's criterion, \cite{Val79_1}]\label{thm:valiant}
Suppose $\phi:\{0, 1\}^\star\to \N$ is a function in $\CountP/\poly$. Then the 
polynomial family $\{f_n\}_{n\in\N}$ defined by 
$f_n=\sum_{\vece\in\{0,1\}^n}\phi(\vece)\vecx^\vece$ is in $\VNP$.
\end{theorem}


\section{Skew Schubert polynomials in $\CountP$ and $\VNP$ }\label{sec:sharp_p}

In this section we first define skew Schubert polynomials via the labeled 
Bruhat order as in \cite{LS03}. We also indicate how Schubert polynomials form a 
special case of skew Schubert polynomials. We then put these polynomials, and 
therefore Schubert polynomials, in $\CountP$ and $\VNP$. 
Also note that it is more convenient to work with permutations instead of codes in this section. 

\paragraph{Definition of skew Schubert polynomials.} The Bruhat 
order on permutations in $S_N$ is defined 
by its covers: for $\sigma, \pi\in S_N$, $\sigma\bcover \pi$ if (i) 
$\sigma^{-1}\pi$ is a transposition $\tau_{ik}$ for $i<k$, and (ii) 
$|\sigma|+1=|\pi|$. Assuming (i), condition (ii) is 
equivalent to: 
\begin{enumerate}
\item[(a)] $\sigma(i) < \sigma(k)$;
\item[(b)] for any $j$ such that $i<j<k$, either $\sigma(j)>\sigma(k)$, or 
$\sigma(j)<\sigma(i)$.
\end{enumerate}
This is because $\pi$ gets an inversion added due to the transposition $\tau_{ik}$.
So in no position between $i$ and $k$ can $\sigma$ take a value between $\sigma(i)$ and $\sigma(k)$.
Else, the number of inversions in $\pi$ will change by more than 1.
Taking the transitive closure gives the Bruhat order ($\leq$). The maximal 
element in Bruhat order is $\pi_0=\underline{N, N-1, \dots, 1}$, whose code is 
$\vecd=(N-1, N-2, \dots, 1)$.

The \emph{labeled Bruhat order} is the key to the definition 
of skew Schubert polynomials. While naming it as an order, it is actually a 
directed graph with multiple labeled edges, with the vertices being the 
permutations in $S_N$. For $\sigma\bcover\pi$ s.t. $\sigma^{-1}\pi=\tau_{st}$, 
$s\leq j< t$ and $b=\sigma(s)=\pi(t)$, add a labeled direct edge as $\sigma 
\xrightarrow{(j, b)} \pi$. That is, for each $\sigma\bcover\pi$, there are $t-s$ edges 
between them. 

For any saturated chain $C$ in this graph, we associate a monomial 
$\vecx^{\vece(C)}$, where $\vece(C)=(e_1, \dots, e_{N-1})$, and $e_i$ counts the 
number of $i$ appearing as the first coordinate of a label in $C$. A chain $$
\sigma_0\xrightarrow{(j_1, b_1)}\sigma_1\xrightarrow{(j_2, 
b_2)}\dots\xrightarrow{(j_m, b_m)} \sigma_m
$$
is \emph{increasing} if its sequence of labels is increasing in the lexicographic 
order on pairs of integers. 

Now we arrive at the definition of skew Schubert polynomials. 
\begin{definition}\label{def:skew}
Let $\vecd$ and $\pi_0$ be as above. Given two 
permutations $\sigma$ and $\pi$, s.t. $\sigma$ is no larger than $\pi$ in the 
Bruhat order ($\sigma\leq \pi$), the \emph{skew Schubert 
polynomial}
\begin{equation}\label{eq:skew}
\schub_{\pi/\sigma}(\vecx):=\sum_{C} \vecx^{\vecd}/\vecx^{\vece(C)},
\end{equation} 
summing over all increasing chains in the labeled Bruhat order from 
$\sigma$ to $\pi$. The \emph{Schubert polynomial}
$\schub_{\sigma}:=\schub_{\pi_0/\sigma}$. 
\end{definition}

Now in the increasing chain (or any chain in the labeled Bruhat order), each edge increases $|\sigma|$ by 1.
So the number of edges in an increasing chain from $\sigma$ to $\pi$ is $|\pi| - |\sigma|$.

\paragraph{Skew Schubert polynomials in $\CountP$ and $\VNP$.} Before describing the 
$\CountP$ algorithm for skew Schubert polynomials, we 
note the following. First, by the correspondence between codes and permutations, 
from a code $\vecv\in\N^n$ we can compute $\langle \vecv\rangle\in S_N$ in time 
polynomial in $|\vecv|$. Also, we have $N=\max_{i\in[n]}\{v_i+i\}\leq |\vecv_i|+n\leq |\vecv_0|+n$. 
Second, the length of the path from $\sigma$ to $\pi$ is $|\pi|-|\sigma|$. 

To start with let us see how \Cref{cor:eval} follows from \Cref{thm:eval}.

\begin{proof}[Proof of \Cref{cor:eval}]
Given $\vecv$, we can compute $\langle \vecv\rangle\in S_N$. Note that $N\leq 
|\vecv|+n$. Then form $\vecw=(N-1, N-2, \dots, 1)$. Invoke \Cref{thm:eval} 
with $(\vecv, \vecw, \veca)$.
\end{proof}

\begin{proof}[Proof of \Cref{thm:eval}] 
Consider the following Turing machine $M$: the input to $M$ are (1) codes $\vecv, \vecw \in \N^m$ in unary; 
(2) $\veca = (a_1, \ldots, a_n) \in \N^n$ in binary; and 
(3) a sequence $\vecs$ of triplets of integers $(s_i,j_i,t_i) \in [N] \times [N] \times [N]$ where $s_i\leq j_i < t_i$, $N= |\vecw|+m$, and 
$i\in[\ell]$, $\ell=|\vecw|-|\vecv|$. 
(Note that all the conditions on $\vecs$ can be checked efficiently.) 
The output of $M$ is a nonnegative integer.

Given the input $(\vecv, \vecw, \veca, \vecs)$, $M$ uses $\vecs$ as a guide to compute an increasing chain from 
$\langle \vecv\rangle$ to $\langle \vecw\rangle$ in the labeled Bruhat order of $S_N$. 
Let $\sigma_0=\langle \vecv\rangle$. 
Suppose the chain to be constructed is $\sigma_0\xrightarrow{(j_1, b_1)} 
\sigma_1\xrightarrow{(j_2, b_2)}\dots\xrightarrow{(j_\ell, b_\ell)} \sigma_\ell$. 
At step $i$, $0\leq i<\ell$, $M$ also maintains an exponent vector $\vece_i\in\N^{N-1}$, and the label $(j_i, b_i)$.

To start, $M$ sets $\vece_0=(0, \dots, 0)$, and $(j_0, b_0)=(0, 0)$ (lexicographically minimal). 
Then when the step $i-1$ finishes, $M$ maintains $\vece_{i-1}$, $\sigma_{i-1}$, and $(j_{i-1}, b_{i-1})$.
Then at step $i$, based on $(s_i, j_i, t_i)$, $M$ performs the following. 
First, it computes $\sigma_i=\sigma_{i-1}\tau_{s_i t_i}$, and checks whether $|\sigma_{i-1}|+1=|\sigma_i|$: 
if equal, then continue; otherwise, return $0$ (not a valid chain). 
Second, it sets $\vece_i$ by adding $1$ to the $j_i$th component of $\vece_{i-1}$, and keeping the other components same. 
Third, it computes $b_i=\sigma_i(t_i)$, and checks whether $(j_i, b_i)$ is larger than $(j_{i-1}, b_{i-1})$ 
in the lexicographic order: if it is larger, then continue; otherwise, return $0$ (not a valid chain).

When the $\ell$th step finishes, $M$ obtains $\sigma_\ell$ and $\vece_\ell$. 
It first checks whether $\sigma_\ell=\langle \vecw\rangle$: if equal, then continue. 
Otherwise, return $0$ (not a valid chain). 
$M$ then computes $\veca^\vecd/\veca^{\vece_\ell}$ as the output. 

This finishes the description of $M$. 
Clearly $M$ runs in time polynomial in the input size. 
$M$ terminates within $\ell$ steps where $\ell=|\vecw|-|\vecv|$, and recall that $\vecw$ and $\vecv$ are given in unary. 

Finally note that $\schub_{\vecw/\vecv}(\veca)$ is equal to $\sum_{\vecs}M(\vecv, \vecw, \veca, \vecs)$, 
where $\vecs$ runs over all sequences of triplets of indices as described at the beginning of the proof. 
By the discussion of $\CountP$ at the end of \Cref{sec:prel}, this puts the 
evaluation of $\schub_{\vecw/\vecv}$ in $\CountP$.
\end{proof}

\begin{proof}[Proof of \Cref{thm:vnp}] Let us outline the proof of \Cref{thm:vnp}, which basically follows from the proof of \Cref{thm:eval}. 
By Valiant's criterion (\Cref{thm:valiant}), to put $\schub_{\vecw/\vecv}$ in $\VNP$, it suffices 
to show that the coefficient of any monomial in $\schub_{\vecw/\vecv}$ is in $\CountP$. 
Therefore we consider the following Turing machine $M'$, which is modified from the Turing machine $M$ in the proof of \Cref{thm:eval} as follows. 
Firstly, $M'$ takes input $(\vecv, \vecw, \veca, \vecs)$, where $\veca$ is thought of as an exponent vector, and is given in unary. 
Second, in the last step, $M'$ checks whether $\vecd-\vece_\ell$ equals $\veca$ or not. 
If equal, then output $1$. Otherwise output $0$. 
It is clear that this gives a $\CountP$ algorithm to compute the coefficient of $\vecx^\veca$. 
The only small problem here is that in the literature (\cite[Prop. 2.20]{B_act}, \cite[Thm. 2.3]{Koi05}) we can find,  
Valiant's criterion is only stated for multilinear polynomials.\footnote{It is 
well-known though that Valiant's criterion works for even non-multilinear 
polynomials. However, the only reference we are aware of is \cite[pp. 10, 
Footnote 1]{Rap_survey}, where no details are provided. Therefore, we describe the 
procedure to overcome this for completeness. }
 But it only takes a little more effort to overcome this; 
essentially, this is because the degree is assumed to be polynomially bounded, so $\veca$ can be given in unary. 
First, note that $N$ (as in the beginning of the proof of \Cref{thm:eval})is an upper bound on the individual degree for each $x_i$. 
Then introduce $N$ copies for each variable $x_i$. 
Since $\veca=(a_1, \dots, a_n)$ is given in unary, we can assume w.l.o.g. that each $a_i$ is an $N$-bit string 
$(a_{i,1}, \dots, a_{i,N})$, and if $a_i=k$, the first $k$ bits are set to $1$, and the rest $0$. 
(These conditions are easy to enforce in the definition of $M'$.) 
We then use $a_{i, j}$ to control whether we have the $j$th copy of $x_i$, or $1$, using the formula $a_{i,j}x_i+1-a_{i,j}$. 
After this slight modification, the Valiant's criterion applies, and the proof is concluded. 
\end{proof}


\section{Sparse interpolation in an interpolation-friendly basis}\label{sec:interpolate}

\subsection{Interpolation-friendly bases}\label{subsec:good}

Let $M=\{t_\lambda\mid \lambda \in \Lambda\}$ be a basis of a $\Z$-module 
$R\subseteq \Z[x_1, \dots, x_n]$, indexed by $\lambda\in\Lambda$. 
Given a function $K:\Lambda \times \N \to \R^+$, $M$ is called \emph{$K$-bounded}, 
if $\forall t_{\lambda}\in M$, the absolute values of the coefficients in the 
monomial expansion of $t_\lambda\in M$, $t_\lambda\in\Z[x_1, \dots, x_n]$ are 
bounded by $K(\lambda, n)$.\footnote{Note that while $\Lambda$ depends on $n$ 
already, we feel that it is clearer to explicitly designate $n$ as a parameter 
of $K$, as seen e.g. in \Cref{prop:schub_bound}.}

For a $(0, 1)$, nonsingular matrix $A$, 
$L_A:=\{\vecc^\star \in (\Q^n)^{\star}\mid \exists \vecc'\in (\Z^+)^n, 
\vecc=A\vecc'\}$. 
Note that $\vecc'$ is a vector with positive integer components. Also recall that 
for $\vecc\in \Q^n$, $\vecc^\star$ denotes the linear form determined by $\vecc$. 
$M$ is \emph{$L_A$-compatible}, if for every $t_\lambda\in M$, we can associate an 
exponent $\vece_\lambda$ s.t. (1) for any $\vecc^\star \in 
L_A$, $\vecc^\star$ achieves the maximum uniquely at $\vece_\lambda$ over 
$E_{t_\lambda}=\{\vece\in\N^n\mid \vecx^\vece\in t_\lambda\}$; 
(2) $\vece_\lambda\neq \vece_{\lambda'}$ for $\lambda\neq 
\lambda'$; (3) the coefficient of $\vecx^{\vece_\lambda}$ in $t_\lambda$ is 
$1$. $\vecx^{\vece_\lambda}$ (resp. $\vece_\lambda$) is called the \emph{leading 
monomial} (resp. \emph{leading exponent}) of $t_\lambda$ w.r.t. $L_A$. By
the conditions (1) and (2), the leading monomials are distinct across 
$M$, and for each $t_\lambda$, the leading monomial is unique. We 
assume that from $\lambda$ it is easy to compute $\vece_\lambda$, and vice versa. 
In fact, for Schubert polynomials, $\lambda$ is from $\Lambda=\Z^n$, 
and $\vece=\lambda$.

Combining the above two definitions, we say $M$ is $(K, 
L_A)$-interpolation-friendly, if 
(1) $M$ is $K$-bounded; (2) $M$ is $L_A$-compatible. We also call it $(K, 
L_A)$-friendly for short, or I-friendly when $K$ and $L_A$ are 
understood from the context.
\begin{itemize}
\item A trivial example is the monomial basis for $\Z[\vecx]$, where $K=1$, $A$ is the identity transformation, and a leading monomial for $\vecx^\vece$ is just itself;
\item For symmetric polynomials, the basis of Schur polynomials (indexed by 
partitions $\alpha$) is $(K, L_A)$-friendly, where (1) $K(\alpha, 
n)=\sqrt{|\alpha|!}$ 
by \Cref{prop:schur_bound}, (2) 
$A=(r_{i,j})_{i,j\in[n]}$ where $r_{i,j}=0$ if $i>j$, and $1$ otherwise, by
the fact that every exponent 
vector in $\schur_\alpha$ is dominated by $\alpha$ (\cite[Sec. 
2.2]{BF97}). ($A$ is the upper triangular matrix with $1$'s on the diagonal 
and above.)
The associated leading monomial for $\schur_\alpha$ is $\vecx^\alpha$.
In retrospect, the fact that Schur polynomials form an I-friendly basis is the 
mathematical support of the Barvinok-Fomin algorithm \cite{BF97}.
\end{itemize}

\subsection{Sparse interpolation in an interpolation-friendly basis}

In this section we perform deterministic sparse polynomial interpolation in an interpolation-friendly basis. 
The idea is to combine the structure of the Barvinok-Fomin algorithm with some ingredients from the Klivans-Spielman algorithm. 

We first briefly review the idea of the Klivans-Spielman algorithm \cite[Section 6.3]{ks01}. 
Suppose we want to interpolate $f\in\Z[x_1, \dots, x_n]$ of degree $d$ with $m$ monomials. 
Their algorithm makes use of the map
\begin{equation}\label{eq:phi}
\phi_{\vecc}(x_1, \dots, x_n)=(y^{c_1}, \dots, y^{c_n})
\end{equation} 
where $\vecc=(c_1, \dots, c_n)\in (\Z^+)^n$.
If $\vecc$ satisfies the property:
$
\forall \vece\neq \vece'\in f, \langle \vecc, \vece\rangle \neq \langle \vecc, \vece'\rangle,
$
then we can reduce interpolation of multivariate polynomials to interpolation of univariate polynomials: 
first apply the univariate polynomial interpolation (based on the Vandermonde matrix) to get a set of coefficients. 
Then to recover the exponents, modify $\phi_\vecc$ as
\begin{equation}\label{eq:phi_prime}
\phi_{\vecc}'(x_1, \dots, x_n)=(p_1 y^{c_1}, \dots, p_n y^{c_n}),
\end{equation}
where $p_i$'s are distinct primes, and get another set of coefficients. 
Comparing the two sets of coefficients we can compute the exponents. 
Note that the components of $\vecc$ need to be small for the univariate interpolation to be efficient. 
To obtain such a $\vecc$, Klivans and Spielman exhibit a small -- polynomial in $n$, $m$, $d$ 
and an error probability $\epsilon\in(0, 1)$ -- set of test vectors $\vecc$ s.t. with probability $1-\epsilon$, a vector from this set 
satisfies the above property. 
Furthermore, the components of these vectors are bounded by $O(m^2nd/\epsilon)$. 
Their construction was reviewed in \Cref{lem:ks}, \Cref{sec:prel}.

Now suppose $M$ is a $(K, L_A)$-friendly basis for $R\subseteq \Q[\vecx]$, and we 
want 
to recover $f=\sum_{\lambda\in \Gamma}a_\lambda t_\lambda$ of degree $\leq d$, and 
$|\Gamma|=m$. To apply the above idea to an arbitrary I-friendly basis $M$, the 
natural 
strategy is to extract the leading monomials w.r.t. $L_A$. However, as each basis 
polynomial can be quite complicated, there are many other non-leading monomials 
which may interfere with the leading ones. Specifically, we need to explain the 
following:
\begin{enumerate}
\item[(1)] Whether extremely large coefficients appear after the map $\phi_\vecc$, 
therefore causing the univariate interpolation procedure to be inefficient?
\item[(2)] Whether some leading monomials are preserved after the map 
$\phi_\vecc$? (That is, will the image of every leading monomial under 
$\phi_\vecc$ be cancelled by non-leading monomials?)
\end{enumerate}

It is immediate to realize that I-friendly bases are designed to overcome the above issues. 
(1) is easy: by the $K$-bounded property, for any monomial $\vecx^\vecu$ in $f$, the absolute value of $\coeff(\vecu, f)$ is bounded by 
$K\cdot (\sum_\lambda |a_\lambda|)$. 
Therefore, the coefficients of the image of $f$ under $\phi_\vecc$ are bounded by $O\big(\binom{n+d}{d} \cdot K\cdot (\sum_\lambda|a_\lambda|)\big)$.
(2) is not hard to overcome either; see the proof of \Cref{thm:interpolate_good} below. 
These properties are used implicitly in the Barvinok-Fomin algorithm. 

There is one final note: if, unlike in the monomial basis case, the procedure cannot produce all leading monomials at one shot, 
we may need to get one $t_\lambda$ and its coefficient, subtract that off, and recurse. 
This requires us to compute $t_\lambda$'s efficiently. 
As this is the property of $t_\lambda$, not directly related to the interpolation problem, 
we assume an oracle which takes an index $\lambda\in\Lambda$ and an input $\veca\in\N^n$, and returns $t_\lambda(\veca)$. 

\begin{theorem}\label{thm:interpolate_good}
Let $M=\{t_{\lambda}\mid \lambda \in \Lambda\}$ be a $(K, L_A)$-friendly basis for 
$R\subseteq \Z[\vecx]$, $K=K(\lambda, n)$, $A$ a $(0, 1)$ invertible matrix. Given 
an access to an 
oracle $\oracle=\oracle(\lambda, \veca)$ that computes basis polynomial 
$t_{\lambda}(\veca)$ for $\veca\in\N^n$, there exists a deterministic algorithm 
that, given a black box containing $f=\sum_{\lambda\in\Gamma}a_\lambda t_\lambda$, 
$a_\lambda\neq 0\in \Z$ with the promise that $\deg(f)\leq d$ and $|\Gamma|\leq 
m$, computes such an expansion 
of $f$ in time $\poly(n, d, m, \log(\sum_\lambda|a_\lambda|), \log K)$.
\end{theorem}

\begin{proof}
We first present the algorithm. Recall the maps $\phi_\vecc$ and $\phi'_\vecc$ 
defined in \Cref{eq:phi} and \Cref{eq:phi_prime}.
\begin{description}
\item[Input:] A black box $\bbox$ containing $f\in R\subseteq \Z[x_1, \dots, x_n]$ 
with the promises: 
(1) $\deg(f) \leq d$; (2) $f$ has $\leq m$ terms in the $M$-basis. 
An oracle $\oracle=\oracle(\lambda, \veca)$ computing $t_{\lambda}(\veca)$ for $\veca\in\N^n$. 

\item[Output:] The expansion $f=\sum_{\lambda\in\Gamma}a_\lambda t_\lambda$.
\item[Algorithm:] \begin{enumerate}
\item By \Cref{lem:ks}, construct the Klivans-Spielman set $\ks=\ks(m, n, 1/3, nd, 
p)$. 
\item For every vector $\vecc$ in $\ks$, do: 
\begin{enumerate}
\item $f_\vecc\leftarrow 0$. $i\leftarrow 0$. $\vecd\leftarrow A\vecc$.
\item While $i< m$, do:
\begin{enumerate}
\item Apply the map $\phi_\vecd$ to $\bbox-f_\vecc$ (with the help of $\oracle$), 
and use the univariate interpolation 
algorithm to obtain $g(y)=\sum_{i=0}^{k}b_iy^i$. If $g(x)\equiv 0$, 
break.
\item Apply the map $\phi_\vecd'$ to $\bbox-f_\vecc$ (with the help of $\oracle$), 
and use the univariate interpolation algorithm to obtain 
$g'(y)=\sum_{i=0}^{k}b_i'y^i$.
\item From $b_k$ and $b_k'$, compute the corresponding monomial $\vecx^\vece$, and 
its coefficient $a_\vece$. 
From $\vecx^\vece$, compute the corresponding label $\nu\in \Lambda$, and set 
$a_\nu\leftarrow a_\vece$. 
\item $f_\vecc\leftarrow f_\vecc+a_\nu t_\nu$. $i\leftarrow i+1$.
\end{enumerate}
\end{enumerate}
\item Take the majority of $f_\vecc$ over $\vecc$ and output it.
\end{enumerate}
\end{description}

We prove the correctness of the above algorithm. 
As before, let $\vece_\lambda$ be the leading vector of $t_\lambda$ w.r.t. $L_A$. 
Note that as $A^{T}$ is $(0, 1)$-matrix, entries in the vector 
$A^{T}\vece_\lambda$ are in $\{0, 1, \dots, nd\}$. By the property of $\ks(m, 
n, 1/3, nd, p)$, no less than $2/3$ fraction of the vectors from $\vecc\in\ks$ 
satisfy that $\langle \vecc, A^{T}\vece_\lambda\rangle$ are distinct over $\lambda\in\Gamma$; call 
these vectors ``distinguishing.'' We shall show that for any distinguishing 
vector, the algorithm outputs the correct expansion, so step (3) would succeed.

Fix a distinguishing vector $\vecc$. As $\vece_\lambda\neq \vece_{\lambda'}$ for 
$\lambda\neq \lambda'$ and $A$ is invertible, there exists a unique 
$\nu\in\Lambda$ s.t. $\vece_\nu$ achieves the maximum at $\langle\vecc, 
A^T\vece_\lambda\rangle$ over $\lambda\in\Gamma$. 
As $\langle \vecc, A^{T}\vece_\lambda\rangle=\langle A\vecc, 
\vece_\lambda\rangle=\langle \vecd, \vece_\lambda\rangle$, 
by the definition of $M$ being $L_A$-compatible, 
we know that within each $t_\lambda$, $\vecd^\star$ achieves the  unique 
maximum at $\vece_\lambda$ over $E_{t_\lambda}$, the set of all exponent vecors in 
$t_\lambda$. Thus over all $\vece\in f$, 
$\vecd^\star$ achieves the unique maximum at $\vece_\nu$. This means that 
$y^{\langle\vecd,\vece_\nu\rangle}$ is the monomial in $g(y)$ of maximum degree, 
and could not be affected by other terms. As 
$\coeff(\vece_\nu, t_\nu)=1$, $\coeff(\vece_\nu, f)$ is just 
the coefficient of $t_\nu$ in the expansion of $f$. So we have justified that from 
Step (2.b.i) to (2.b.iii), the algorithm extracts the monomial of maximum degree 
in $g(y)$, 
computes the corresponding coefficient and exponent in $f$, and interprets as a 
term $a_\lambda t_\lambda$.

To continue computing other terms, we just need to note that $\vecc$ is still 
distinguishing 
w.r.t. ($f-$ some of the terms within $f$). This justifies Step (2.b.iv).

To analyze the running time, the FOR-loop in Step (2) and the WHILE-loop 
in Step (2.b) take $O(m^2n)$ and $m$ rounds, respectively. In the univariate 
polynomial interpolation step, as the components in 
$\vecc$ are bounded by $O(m^2n\cdot nd)$ and $A$ is $(0, 1)$, the components in 
$\vecd$ are bounded by 
$O(m^2n^3d)$. It follows that $k=\deg(g)=\langle\vecd, 
\vece_\nu\rangle=O(m^2n^3d^2)$. By the $K$-bounded property, the coefficients of 
$g(y)$ are of magnitude $O\big(\binom{n+d}{d}\cdot K\cdot 
(\sum_\lambda|a_\lambda|)\big)$. So the running time 
for the univariate interpolation step, and therefore for the whole algorithm, 
is $\poly(m, n, d, \log(\sum_\lambda|a_\lambda|),\log K)$. 
\end{proof}


\section{Schubert polynomials form an interpolation-friendly basis}\label{sec:schub_good}

In this section, our ultimate goal is to prove \Cref{prop:schub_bound} and 
\Cref{prop:schub_friendly}, which establish that Schubert polynomials form an 
I-friendly 
basis. The main theorem \Cref{thm:interpolate} follows immediately. For this, we 
need to review some properties of Schur polynomials, 
the definition of Schubert polynomials via divided differences, and the transition 
formula \cite{LS85}. The transition formula is the main technical tool to deduce 
the properties of Schubert polynomials we shall need for \Cref{prop:schub_bound} 
and \Cref{prop:schub_friendly}. These include \Cref{lem:dom} which helps us to find 
the matrix $A$ needed for the $L_A$-compatible property, and an alterative proof 
for Schubert polynomial in $\CountP$.

\paragraph{Schur polynomials.} For a positive integer $\ell$, the complete 
symmetric polynomial 
$h_\ell(\vecx)\in\Z[\vecx]$ is the sum over all monomials of degree $\ell$, with 
coefficient $1$ for every monomial. We also define $h_0(\vecx)=1$, and 
$h_\ell(\vecx)=0$ for any $\ell < 0$. For a partition $\alpha=(\alpha_1, \dots, 
\alpha_n)$ in $\N^n$, the Schur polynomial $\schur_\alpha(\vecx)$ in $\Z[\vecx]$ 
can be defined by the Jacobi-Trudi formula as 
$\schur_\alpha(\vecx)=\det[h_{\alpha_i-i+j}(\vecx)]_{i, j\in[n]}.$
Note that $\deg(\schur_\alpha(\vecx))=|\alpha|$.
Via this determinantal expression, we have
\begin{proposition}[{\cite[Sec. 2.4]{BF97}}]\label{prop:compute_schur}
For $\veca\in\Z^n$, $\schur_\alpha(\veca)$ can be computed using 
$O(|\alpha|^2\cdot n + n^3)$ 
arithmetic operations, and the bit lengths of intermediate numbers are polynomial 
in those of $\veca$.
\end{proposition}
Littlewood's theorem shows Schur polynomials have positive 
coefficients. We also need the following bound on coefficients -- the Kostka 
numbers -- in $\schur_\alpha(\vecx)$.
\begin{proposition}[{\cite[Sec. 2.2]{BF97}}]\label{prop:schur_bound}
For any $\vece\in\N^n$, $0\leq \coeff(\vece, \schur_\alpha(\vecx))\leq 
\sqrt{|\alpha|!}$.
\end{proposition}

\paragraph{Definition of Schubert polynomials via divided differences.} 

We follow the approach in \cite{Lascoux_symmetric,Lascoux_parallel}.
For $i\in[n-1]$, let 
$\chi_i$ be the switching operator on $\Z[x_1, \dots, x_n]$: 
$$f^{\chi_i}(x_1, \dots, 
x_i, x_{i+1}, \dots, x_n):=f(x_1, \dots, x_{i+1}, x_i, \dots, x_n).
$$
Then the divided difference operator $\partial_i$ on $\Z[x_1, \dots, x_n]$ is 
$\partial_i(f):= \frac{f-f^{\chi_i}}{x_i-x_{i+1}}.$

\begin{definition}\label{def:schub}
For $\vecv=(v_1, \dots, v_n)\in \N^n$, the Schubert polynomial $Y_\vecv\in \Z[x_1, 
\dots, x_n]$ is defined recursively as follows:
\begin{enumerate}
\item If $\vecv$ is dominant, then $Y_\vecv=x_1^{v_1} x_2^{v_2}\dots x_n^{v_n}$.
\item If $v_i>v_{i+1}$, then $Y_{\vecv'}=\partial_i Y_\vecv$ where 
$\vecv'=(v_1, \dots, v_{i+1}, v_i-1, \dots, v_n)$.
\end{enumerate}
\end{definition}

It is not hard to see that this defines $Y_{\vecv}$ for any $\vecv$.
We list some basic facts about Schubert polynomials.
\begin{fact}[{\cite{Manivel}}]\label{fact:schub}
\begin{enumerate}
\item If $\vecv=(v_1, \dots, v_n)$ is anti-dominant with $k$ being the last 
nonzero index, $Y_\vecv$ equals the Schur polynomial $\schur_\alpha(x_1, \dots, 
x_k)$ where $\alpha=(v_k, \dots, v_1)$.
\item As a special case of (1), if $\vecv=(0, \dots, 0, w, 0, \dots 0)$ where $w$ 
is at the $k$th position, 
then $Y_\vecv(\vecx)$ is the complete homogeneous symmetric polynomial 
$h_w(x_1, \dots, x_k)$.
\item If $\vecv=[v_1, \dots, v_k, 0, \dots, 0]\in \N^n$, then 
$Y_{\vecv}\in\Z[x_1, \dots, x_k]$.
\end{enumerate}
\end{fact}

\paragraph{The transition formula and its applications.} 
Given a code $\vecv\in\N^n$, let 
$k$ be the largest $t$ s.t. $v_t$ is nonzero, and $\vecv'=[v_1, \dots, v_k-1, 0, 
\dots, 0]$. For convenience let $\sigma=\aperm{\vecv'}$.
Then the transition formula of Lascoux and Sch\"utzenberger \cite{LS85} is:
$$Y_\vecv=x_k Y_{\vecv'}+ \sum_{\vecu} Y_{\vecu},$$
where $\vecu\in \N^n$ satisfies that: (i) $\aperm{\vecu}\sigma^{-1}$ is a 
transposition $\tau_{ik}$ for $i<k$; (ii) $|\vecu|=|\vecv|$. Assuming (i), 
condition (ii) is 
equivalent to: 
\begin{multline}\label{eqn:trans}
\text{(a) } \sigma(i) < \sigma(k); 
 \text{(b) for any } j \text{ s.t. } i<j<k, \\
\text{ either } \sigma(j)>\sigma(k), 
\text{ or } \sigma(j)<\sigma(i).
\end{multline}
Let $\Psi_\vecv$ be the set of codes with weight 
$|\vecv|$ appearing in the transition formula for $\vecv$, and 
$\Phi_{\vecv}=\Psi_{\vecv}\cup \{\vecv'\}$. Any $\vecu\in 
\Psi_\vecv$ is uniquely determined by the transposition $\tau_{ik}$, therefore, by 
some $i\in [k-1]$. 

The transition formula yields the following simple, yet rarely 
mentioned\footnote{The only reference we know of is in 
\cite[pp. 62, Footnote 4]{Lascoux_polynomial}. } property of Schubert 
polynomials. 
This is the key to show that the Schubert basis is $L_A$-compatible for some 
appropriate $A$. For completeness we include a proof here. 

Given $\vecv$ and $\vecu$ in $\N^n$, $\vecv$ dominates $\vecu$ reversely, denoted 
as $\vecv\dom\vecu$, if 
$v_n\geq u_n$, $v_n+v_{n-1}\geq 
u_n+u_{n-1}$, \dots, $v_n+\dots +v_2\geq u_n+\dots +u_2$, 
$v_n+\dots+v_1=u_n+\dots+u_1$. 

\begin{lemma}\label{lem:dom}
For $\vecu\in \N^n$, if $\vecx^{\vecu}$ is in $Y_\vecv$, then $\vecv\dom \vecu$. 
Furthermore, $\coeff(\vecv, \schub_\vecv)=1$.
\end{lemma}
\begin{proof} 
We first induct on the weight. When the weight is $1$, the claim holds 
trivially. Assume the claim holds for weight $\leq w$, and consider 
$\vecv\in\N^n$ with $|\vecv|=w+1$. We now induct on the reverse dominance order, 
from small to large. The smallest one with weight $w+1$ is $[w+1, 0, \dots, 0]$. 
As $Y_{[w+1, 0, \dots, 0]}=x_1^{w+1}$, the claim holds.

For the induction step, we make use of the transition formula. Suppose
$k$ is the largest $i$ s.t. $v_i$ is nonzero, $\vecv'=[v_1, \dots, v_k-1, 0, 
\dots, 0]$, and $\sigma=\aperm{\vecv'}$.
Then by the transition formula, $Y_\vecv=x_k Y_{\vecv'}+ \sum_{\vecu} Y_{\vecu}$, 
where $\vecu\in \N^n$ 
satisfies that: (i) $\aperm{\vecu}\sigma^{-1}$ is a 
transposition $\tau_{ik}$ for $i<k$; (ii) $|\vecu|=|\vecv|$. Assuming (i), the 
condition (ii) is 
equivalent to that: (a) $\sigma(i) < \sigma(k)$; (b) for any $j$ s.t. $i<j<k$, 
either $\sigma(j)>\sigma(k)$, or $\sigma(j)<\sigma(i)$. Thus $\vecu$ and $\vecv'$ 
can only differ at positions $i$ and $k$, and $u_i > v'_i=v_i$, $u_k\leq 
v'_k<v_k$. It follows 
that $\vecv\dom\vecu$, and it is clear that $\vecv\neq \vecu$. By the induction 
hypothesis on $|\vecv|$, each monomial 
in $Y_{\vecv'}$ is reverse dominated by $\vecv'$, and $\coeff(\vecv', 
\schub_{\vecv'})=1$. As $Y_{\vecv'}$ depends only 
on $x_1, \dots, x_k$ by \Cref{fact:schub} (3), each monomial in 
$x_kY_{\vecv'}$ is reverse dominated by $\vecv$. By the induction 
hypothesis on the reverse dominance order, every monomial in $Y_\vecu$ is  
reverse dominated by $\vecu$, thus is reverse dominated by $\vecv$ and cannot be 
equal to $\vecv$. Thus 
$\coeff(\vecv, \schub_\vecv)=1$, which is from $x_k\schub_{\vecv'}$. This finishes 
the induction step.
\end{proof}

We then deduce another property of Schubert polynomials from the transition 
formula. Starting with $\vecv_0$, we can form a chain of transitions 
$
\vecv_0 \to \vecv_1 \to \vecv_2 \to \dots \to \vecv_i \to \dots
$
where $\vecv_{i}\in \Psi_{\vecv_{i-1}}$. The following lemma shows that 
long enough transitions lead to anti-dominant codes.
\begin{lemma}[{\cite[Lemma 3.11]{LS85}}]\label{lem:long_trans}
Let $\vecv_0\to \vecv_1\to \dots \to \vecv_\ell$ be a sequence of codes in $\N^n$, 
s.t. $\vecv_{i}\in \Psi_{\vecv_{i-1}}$, $i\in[\ell]$. If none of $\vecv_i$'s are 
anti-dominant, then $\ell \leq n\cdot |\vecv_0|$. 
\end{lemma}

Based on \Cref{lem:long_trans}, we have the following corollary.
Recall that for $\vecv\in\N^n$, $\Phi_\vecv$ is the 
collection of codes (not necessarily of weight $|\vecv|$) in the transition 
formula for $\vecv$. 
\begin{corollary}\label{cor:long_trans}
Let $\vecv_0\to \vecv_1\to \dots \to \vecv_\ell$ be a sequence of codes in $\N^n$, 
s.t. $\vecv_{i}\in \Phi_{\vecv_{i-1}}$, $i\in[\ell]$. If none of $\vecv_i$'s are 
anti-dominant, then $\ell \leq n\cdot (|\vecv_0|^2+|\vecv_0|)$. 
\end{corollary}

\begin{proof} 
For $w\in[|\vecv_0|]$, let $i_w$ be the last index $i$ in $[\ell]$ s.t. $\vecv_i$ 
is of weight $w$. As none of $\vecv_i$'s are anti-dominant, by 
\Cref{lem:long_trans}, $i_{w-1}-i_w\leq  n\cdot w+1 \leq n\cdot |\vecv_0|+1$. 
The result then follows. 
\end{proof}
In fact, by following the proof of \Cref{lem:long_trans} as in \cite{LS85}, 
it is not hard to show that in \Cref{cor:long_trans} the same bound as in \Cref{lem:long_trans}, namely $\ell \leq n \cdot |\vecv_0|$, holds. 

From \Cref{cor:long_trans}, a $\CountP$ algorithm for Schubert polynomials can also be derived. 

\begin{proof}[An alternative proof of \Cref{cor:eval}]
Consider the following Turing 
machine $M$: the input to $M$ is a code $\vecv\in\N^n$ in unary, a point 
$\veca=(a_1, \dots, a_n)\in\N^n$ in binary, and a sequence $\vecs$ of pairs of 
indices $(s_i, t_i)\in 
[n]\times [n]$, $s_i\leq t_i$, and $i \in [\ell]$ where $\ell:=n\cdot 
(|\vecv|^2+|\vecv|)$. 
The output of $M$ is a 
nonnegative integer. Given the input $(\vecv, \veca, \vecs)$, $M$ computes a 
sequence 
of codes $\vecv_0\to\vecv_1\to\dots\to\vecv_\ell$, and keeps track of a 
monomial $\vecx^{\vece_0}\to\vecx^{\vece_1}\to\dots\to\vecx^{\vece_\ell}$, where 
$\vece_i\in\N^n$. The pair 
of indices $(s_{i+1}, t_{i+1})$ is used as the instruction to 
obtain $\vecv_{i+1}$ from $\vecv_i$, and $\vece_{i+1}$ from $\vece_i$. 

To start, $\vecv_0=\vecv$, and $\vece=(0, \dots, 0)$. Suppose at step $i$, 
$\vece_i=(e_1, \dots, e_n)$, and $\vecv_i=(v_1, \dots, v_k, 0, \dots, 0)$, 
$v_k\neq 0$. ($k$ is the maximal nonzero index in $\vecv_i$.)

If $\vecv_i$ is anti-dominant, then $Y_{\vecv_i}$ equals to some Schur polynomial 
by \Cref{fact:schub} (1). Using \Cref{prop:compute_schur} $M$ can 
compute the evaluation of that Schur polynomial on $\veca$ efficiently,
and then multiply with the value $\prod_{i\in[n]}a_i^{e_i}$ as the output. Note 
that as will be seen below, the weight of $\vece_i$ is $\ell\leq n\cdot 
|\vecv_0|^2$, so the bit length of $\prod_{i\in[n]}a_i^{e_i}$ is polynomial in the input size. 

In the following $\vecv_i$ is not anti-dominant. $M$ checks whether $t_{i+1}= k$. 
If not, $M$ outputs $0$. 

In the following $t_{i+1}=k$. $M$ then checks whether $s_{i+1}=t_{i+1}$. 

If $s_{i+1}=t_{i+1}$, $M$ goes to step $i+1$ by setting $\vecv_{i+1}=(v_1, \dots, v_k-1, 0, \dots, 0)$, 
and $\vece_{i+1}=(e_1, \dots, e_{k-1}, e_k+1, e_{k+1}, \dots, e_n)$.

If $s_{i+1}<t_{i+1}$, then $M$ tests whether $s_{i+1}$ is an index in $\Psi_{\vecv_i}$, as follows. 
It first computes the permutation $\sigma:=\langle (v_1, \dots, v_{k-1}, v_k-1, 0, \dots, 0)\rangle\in S_N$, 
using the procedure described in \Cref{sec:prel}. 
Note that $N= \max_{i\in[n]}\{v_i+i\}\leq |\vecv_i|+n\leq |\vecv_0|+n$. 
Then it tests whether $s_{i+1}$ is in $\Psi_{\vecv_i}$, using \Cref{eqn:trans}. 
If $s_{i+1}$ is not in $\Psi_{\vecv_i}$, then $M$ outputs $0$. 
Otherwise, $M$ goes to step $i+1$ by setting $\vecv_{i+1}=\acode{\sigma\tau_{s_{i+1}, t_{i+1}}}$, and $\vece_{i+1}=\vece_i$.

This finishes the description of $M$. 
Clearly $M$ runs in time polynomial in the input size. 
$M$ terminates within $\ell$ steps by \Cref{cor:long_trans}. 
$M$ always outputs a nonnegative integer as Schur polynomials are polynomials with positive coefficients.

Finally note that $\schub_\vecv(\veca)$ is equal to $\sum_\vecs M(\vecv, \veca, \vecs)$, where $\vecs$ runs over all the sequences of pairs of 
indices as described at the beginning of the proof. 
By the discussion on $\CountP$ in \Cref{sec:prel}, this puts evaluating 
$\schub_\vecv$ on $\veca$ in $\CountP$.
\end{proof}

\paragraph{The Schubert basis is interpolation-friendly.} Now we are in the 
position to prove that the Schubert basis is 
interpolation-friendly.

\begin{proposition}\label{prop:schub_bound}
$\{Y_{\vecv}\in\Z[x_1, \dots, x_n] \mid \vecv\in\N^n\}$ is $K$-bounded for 
$K(\vecv, n)=n^{2n\cdot (|\vecv|^2+|\vecv|)}\cdot \sqrt{|\vecv|!}$.
\end{proposition}

\begin{proof} 
The alternative proof of \Cref{cor:eval} for Schubert polynomials implies that $\schub_\vecv$ 
can be written as a sum of at most $(n^2)^{n\cdot (|\vecv|^2+|\vecv|)}$ 
polynomials $f$, 
where $f$ is of the form $\vecx^\vece\cdot \schur_\alpha$, $|\alpha|+|\vece|=|\vecv|$. 
The coefficients in Schur polynomial of degree $d$ are bounded by $\sqrt{d!}$ by \Cref{prop:schur_bound}. 
The claim then follows.
\end{proof}

\begin{proposition}\label{prop:schub_friendly}
$\{Y_{\vecv}\in\Z[x_1, \dots, x_n] \mid \vecv\in\N^n\}$ is $L_A$-compatible for 
$A=(r_{i,j})_{i, j\in[n]}$, $r_{i,j}=0$ if $i < j$, and $1$ otherwise. The 
leading monomial of $Y_{\vecv}$ w.r.t. $L_A$ is $\vecx^\vecv$.
\end{proposition}
The matrix $A$ in \Cref{prop:schub_friendly} is the lower triangular matrix 
of $1$'s on the diagonal and below. Compare with that for Schur polynomials, 
described in \Cref{subsec:good}. 

\begin{proof} 
This follows easily from \Cref{lem:dom}: note that for any 
$\vecc=(c_1, \dots, c_n)\in(\Z^+)^n$, $\vecu\in\schub_{\vecv}$, $\langle A\vecc, 
\vecu\rangle=c_1(u_1+\dots + u_n)+c_2(u_2+\dots+u_n)+\dots+c_n u_n\leq 
c_1(v_1+\dots+v_n)+c_2(v_2+\dots+v_n)+\dots+c_nv_n=\langle A\vecc, \vecv\rangle$. 
As $c_i>0$, the equality holds if and only if $\vecu=\vecv$.
\end{proof}

Now we conclude the article by proving the main \Cref{thm:interpolate}. 
\begin{proof}[Proof of \Cref{thm:interpolate}]
Note that $n^{2n\cdot(|\vecv|^2+|\vecv|)}\cdot \sqrt{|\vecv|!}$ is upper bounded 
by 
$2^{O(n\log n\cdot 
(|\vecv|^2+|\vecv|)+|\vecv|\log(|\vecv|))}$, and recall $|\vecv|=\deg(Y_\vecv)$. 
Then 
combine \Cref{prop:schub_bound}, \Cref{prop:schub_friendly}, and 
\Cref{thm:interpolate_good}.
\end{proof}



\paragraph{Acknowledgement.}
  Part of the work was done when 
Youming was 
visiting 
the Simons Institute for the program Algorithms and Complexity in Algebraic 
Geometry. We are grateful to Allen Knutson for his answer 
at \url{http://mathoverflow.net/q/186603}. We would like to thank the anonymous 
reviewers whose suggestions help to improve the writing of this paper greatly.   
Youming's research was supported by 
Australian Research Council DE150100720. 
Priyanka's research was supported by core grants for Centre for Quantum Technologies, NUS.


\bibliographystyle{alpha}
\bibliography{ref}

\appendix

\section{An alternative proof of skew Schubert 
polynomials in $\VNP$.}

By Equation \Cref{eq:skew}, skew Schubert
polynomial can be written as:
\begin{multline}
  \schub_{\pi/\sigma}(\vecx) = \schub_{\pi/\sigma}(x_1, x_2, \ldots x_N) \\ :=
  \sum_{C} \vecx^{\vecd}/\vecx^{\vece(C)} 
  = \sum_{C} \prod_{i=1}^{N-1} x_i^{N - i - e_i} 
 \label{eq:skew2}
\end{multline}
Note that $e_N = 0$.

Let $m$ be the length of an increasing chain. 
The first and second indices of each label in an increasing chain are encoded by 
$N \times m$ $0-1$ matrices $g$ and $b$
respectively.
In these matrices the variables are listed along the row and the edges of a chain 
along the column.
So for each column, there will be a 1 in the row which corresponds to the index that appears in that label.
Thus in case of the $g$ matrix it indicates which variable should be multiplied in the monomial.

All permutations are represented by $N \times N$ $0-1$ matrices, acting on  
length-$N$ column vectors.
$W_0$ and $V$, representing permutations $\sigma$ and $\pi$ respectively are given.
Let $W_1, \ldots W_m$ be the intermediate permutations in the increasing chain.

Now consider the following polynomials.

\begin{eqnarray}
 h_{1N} = \prod_{i=1}^{N} x_i^{(N-i)} \prod_{j=1}^m (\sum_{k=1}^{N} x_k^{-1} 
 g_{kj} )
 \label{poly:h_1n}
\end{eqnarray}
$h_{1N}$ encodes the monomial for a given increasing chain, which depends on $g$. 

In the following we shall gradually build up a series of polynomials, which 
basically characterize the property that $g$ and $b$ form an increasing chain. 

\begin{eqnarray}
 h_{2N} = \prod_{t=1}^m [ ( \prod_{i,j,l,k} ( 1 - (W_t)_{ij} (W_t)_{lk} ) 
 )\cdot(\prod_{i=1}^N \sum_{j=1}^N (W_t)_{ij} ) ]
 \label{poly:h_2n}
\end{eqnarray}
where the second product is over all $1 \leq i,j,l,k \leq N$ such that $i = l$ iff 
$j \neq k$. 
$h_{2N}$ encodes valid permutation matrices. 
That is, it is non-zero iff $W_1, \ldots W_m$ are valid permutation matrices, 
that is each row and column contains exactly one 1.

\begin{eqnarray}
 h_{3N} = [\prod_{i,j,k} (1 - g_{ik} g_{kj} ) ]\cdot[\prod_{i,j,k} (1 - b_{ik} 
 b_{kj} ) ]
 \label{poly:h_3n}
\end{eqnarray}
where both the products are over all $ 1 \leq k \leq m$ and $1 \leq i,j \leq N$ 
such that $i \neq j$. 
$h_{3N}$ is non-zero iff each column of $g$ and $b$ has at most one 1.

\begin{eqnarray}
 h_{4N} = [\prod_{j=1}^m \sum_{k=1}^N g_{kj} ]\cdot[\prod_{j=1}^m \sum_{k=1}^N 
 b_{kj} ]
 \label{poly:h_4n}
\end{eqnarray}
$h_{4N}$ is non-zero iff there is at least one 1 in each column of $g$ and $b$.

\begin{eqnarray}
 h_{5N} = \prod_{i,j=1}^N [1 - ( (W_m)_{ij} -V_{ij} ) ] \cdot [1 + ( (W_m)_{ij} 
 -V_{ij} ) ]
 \label{poly:h_5n2}
\end{eqnarray}
$h_{5N}$ is non-zero iff $W_m = V$.

\begin{multline}
 h_{6ijt} 
 = \prod_{a,b = 1}^N [1 - ((W_t)_{ab} - (W_{t-1} \tau_{ij})_{ab} ) ].
 [1 + ((W_t)_{ab} - (W_{t-1} \tau_{ij})_{ab} ) ]\cdot b_{jt}
 \label{poly:h_6ijt}
\end{multline}
$h_{6ijt}$ is non-zero iff for a particular transposition $(i,j)$ and for a pair of consecutive permutations 
$W_{t-1}$ and $W_t$, (a) $W_t = W_{t-1} \tau_{ij}$ and (b) second index of the label is given according to definition of 
labeled Bruhat order.

\begin{eqnarray}
 h_{7t} = \sum_{i,j=1}^N \sum_{\substack{k,l=1 \\ k < l}}^N (W_{t-1})_{ik} \cdot 
 (W_{t-1})_{jl} \cdot 
	  \prod_{\substack{i<a<j \\ k<b<l}} [1 - (W_{t-1})_{ab} ] \cdot h_{6ijt}
 \label{poly:h_7t}
\end{eqnarray}
$h_{7t}$ is non-zero iff for any pair of consecutive permutations $\sigma' , \pi'$ being encoded by 
$W_{t-1}$ and $W_t$ respectively such that $\pi' = \tau_{ij} \sigma'$ ,
we have $\sigma'(i) < \sigma'(j)$ ($i < j$) and for every $i < k < j$, either $\sigma'(k) < \sigma'(i)$ or 
$\sigma'(k) > \sigma'(j)$.
That is, $\sigma' < \pi'$ in the Bruhat order.

\begin{eqnarray}
 h_{8N} = \prod_{t=1}^m h_{7t} \sum_{i \leq s < j} g_{st}
 \label{poly:h_8n}
\end{eqnarray}
where $1 \leq s < N $. 
$h_{8N}$ is non-zero iff the sequence of permutations is correct, namely, 
maintaining the Bruhat order so that the labels are given accordingly.

\begin{eqnarray}
 h_{9N} = \prod_{t=1}^{m-1} \sum_{i=1}^N g_{it} [\sum_{j=i+1}^N g_{j,t+1} + 
 g_{i,t+1} (\sum_{k=1}^N b_{kt} \cdot
  \sum_{l=k}^N b_{l,t+1} )   ]
  \label{poly:h_9n}
\end{eqnarray}
$h_{9N}$ is non-zero iff each pair of labels in a chain respects the increasing lexicographic order.

We define the polynomial :
\begin{multline}
 h_N(x_1, \ldots x_N, g, b, W_1, \ldots W_m) \\ = h_{1N}\cdot h_{2N}\cdot 
 h_{3N}\cdot 
 h_{4N}\cdot h_{5N}\cdot h_{8N}\cdot h_{9N}
 \label{poly:vp}
\end{multline}
For each assignment to $g, b, W_i$, $h_N$ is either $0$, or a monomial 
corresponding to one correct increasing chain.
It is clear that the size of a straight line program to evaluate $h_N$ is polynomial in $N$.
So $h_N$ is in $\VP$.

The skew Schubert polynomial, then can be given by
\begin{eqnarray}
 \schub_{\pi / \sigma}(x_1, \ldots x_N) = \sum_{g, b, W_1, \ldots W_m} h_N(x_1, 
 \ldots x_N, g, b, W_1, \ldots W_m)
 \label{poly:schub}
\end{eqnarray}
$g$ and $b$ have $mN$ elements and since $m$ is $\mathcal{O}(N^2)$ so each has $\mathcal{O}(N^3)$ elements.
Each $W_i$ has $N^2$ entries.
Thus the summation is over 0-1 strings of polynomial length.

This proves $\schub_{\pi/\sigma}(\vecx)$ is p-definable and hence is in $\VNP$.

\end{document}